\documentclass[12pt]{iopart}

\usepackage{iopams}
\usepackage{harvard}
\usepackage{bibtexlogo}
\usepackage{amsthm}

\newtheorem{theorem}{Theorem}[section]

\newtheorem{define}[theorem]{Definition}

\newcommand{\wh}{Weyl-Heisenberg group }
\newcommand{\fg}{$\mathbf{f}_n(\Phi)$ and $\mathbf{g}_n(\Phi)$ }

\begin{document}
\bibliographystyle{dcu}

\title[]{Group Theoretical Classification of SIC-POVMs}

\author{S. B. Samuel and Z. Gedik}

\address{Faculty of Engineering and Natural Sciences, Sabanci University, Tuzla, Istanbul 34956, Turkey}
\ead{solomon@sabanciuniv.edu , gedik@sabanciuniv.edu}
\vspace{10pt}
\begin{indented}
\item[]October 2023
\end{indented}

\begin{abstract}
The Symmetric Informationally Complete Positive Operator-Valued Measures (SIC-POVMs) are known to exist in all dimensions $\leq 151$ and few higher dimensions as high as $1155$. All known solutions with the exception of the Hoggar solutions are covariant with respect to the Weyl-Heisenberg group and in the case of dimension 3 it has been proven that all SIC-POVMs are Weyl-Heisenberg group covariant. In this work, we introduce two functions with which SIC-POVM Gram matrices can be generated without the group covariance constraint.  We show analytically that the SIC-POVM Gram matrices exist on critical points of surfaces formed by the two functions on a subspace of symmetric matrices and we show numerically that in dimensions 4 to 7, all SIC-POVM Gram matrices lie in disjoint solution "islands". We generate $O(10^6)$ and $O(10^5)$ Gram matrices in dimensions 4 and 5, respectively and $O(10^2)$ Gram matrices in dimensions 6 and 7. For every Gram matrix obtained, we generate the symmetry groups and show that all symmetry groups contain a subgroup of $3n^2$ elements. The elements of the subgroup correspond to the Weyl-Heisenberg group matrices and the order-3 unitaries that generate them. All constructed Gram matrices have a unique generating set. Using this fact, we generate permutation matrices to map the Gram matrices to known Weyl-Heisenberg group covariant solutions. In dimensions 4 and 5, the absence of a solution with a smaller symmetry, strongly suggests that non-group covariant SIC-POVMs cannot be constructed.
\end{abstract}

%
%
%
%
%

\section{Introduction}

Measurements in quantum mechanics are mathematically represented by a set of positive operators known as the positive operator-valued measures (POVMs). The elements of a POVM $E_k$ satisfy the condition $ \sum{E_k} = \mathbb{I} $. For a density matrix $\rho$, the corresponding probability that the $k^{th}$ measurement is obtained is given by $Tr(E_k\rho)$. An informationally complete POVM allows us to reconstruct an arbitrary density matrix from the probabilities of the measurement \cite{prugovevcki1977information,busch1991informationally,nielsen2010quantum}. The most symmetric type of POVM is the Symmetric Informationally Complete POVM, or the  SIC-POVM. The Symmetry in its name refers to the property that the all element of the SIC-POVM are equidistant from one another, i.e., $Tr(E_j E_k) = \frac{n \delta_{jk}+1}{n^2(n+1)}$. This property makes the SIC-POVMs ideal for quantum state tomography by minimizing the total number of measurements needed \cite{scott2006tight}. In addition to tomography SIC-POVMs have applications in Quantum cryptography \cite{durt2008wigner} and the foundational study of QBism \cite{fuchs2010qbism,fuchs2017sic}.

The SIC-POVMs have been constructed in many dimensions using both analytic and numerical methods. The analytic method was first proposed by Zauner in 1999. Zauner showed that a group covariant SIC-POVM can be constructed by using the Weyl-Heisenberg group  \cite{zauner} and constructed solutions in dimensions 2-7. In the same work, Zauner conjectured that a fiducial vector exists in all dimensions, where it is an eigenvector of a special matrix known as the Zauner matrix. A few years later, Renes \cite{renes2004symmetric} independently proposed the Weyl-Heisenberg group for construction of a fiducial vector and explored numerical solutions and computed the complete list of Weyl-Heisenberg group covariant SIC-POVMs in dimensions $\leq 7$. Through further researches on the extended Clifford group, Appleby \cite{appleby2005symmetric,appleby2011lie} and Zhu \cite{zhu2010sic} showed more symmetries of fiducial vectors which simplified the search further.  

The SIC-POVMs are a special type of frames, and frames in general are best described using their Gram matrices \cite{waldron}. For a SIC-POVM, the Gram matrix is formed from the inner products of the vector elements of the SIC-POVM. The Gram matrix allows for the representation of SIC-POVMs up to a unitary equivalence, making it ideal for the study of general solutions. We use the Gram matrix to define two necessary and sufficient conditions to define a SIC-POVM Gram matrix. By not considering the \wh covariance, the number of free parameters becomes of the order of $ n^4 $ which is much larger compared to the Welch bound, but it converges to a solution faster in the dimensions 4-7. 

We derive two functions of the Gram matrix with which we will construct numerical SIC-POVM Gram matrices. In order to characterize the different solutions, we explore symmetries in the functions both analytically and numerically. Analytically, we describe the two trivial symmetries of the equations. Numerically, we show that the functions only have the no other more symmetry other than the trivial symmetries. The functions allow us to treat the SIC-POVM Gram matrices as intersections of two surfaces. We use these to show that SIC-POVMs exist on critical points of both functions and by using the Hessian matrix of the functions, show that all generated solutions are isolated in dimensions 4-7 in agreement with  \cite{bruzda2017quantum}. Lastly, we characterize the Gram matrices we generate the generating set (see \ref{generating set}) associated with the Gram matrix and present the unique generating sets fond in dimensions 4 and 5. In dimensions 6 and 7, the set of numerical solutions we generate are not large enough to make the same claim as the former, but the results in these dimensions have similar properties as dimensions 4 and 5. In dimension 3, the equivalence off all solutions to \wh covariant SIC-POVMs was analytically proven by Hughston \cite{hughston2016surveying}.

\subsection{Construction of the Weyl-Heisenberg Group Covariant SIC-POVMs}

The construction of SIC-POVMs is, in general, a computationally intensive process. This is due to the fact that a general SIC-POVM is defined by $2n^3$ parameters. The number of free parameters is greatly reduced if the SIC-POVM is group covariant. A \wh covariant SIC-POVM only has $2n$ free parameters. This number is further reduced with the symmetries described by Zauner \cite{zauner} and Appleby \cite{appleby2005symmetric}. In low dimensions, group covariant solutions can be derived analytically, but the difficulty increases sharply with increasing dimension. therefore, a numerical method is used to construct solutions in higher dimensions. With enough precision and the symmetries of the solution, algebraic form of a solutions can be constructed, making the numerical method an integral tool in the study of the SIC-POVMs.

The Weyl-Heisenberg group is the generalization of the Pauli group generated by the shift operator and the phase operator, $\hat{X}=\sum_{k=0}^{n-1}{|k+1\rangle \langle k|}$ and $\hat{Z}=\sum_{k=0}^{n-1}{e^{i\frac{2\pi k}{n}}|k\rangle\langle k|}$, respectively. We then define elements of the Weyl-Heisenberg group as shown in equation \ref{W-H group}. Since the Weyl-Heisenberg group is a representation of the $\mathbb{Z}_n \times \mathbb{Z}_n$ group we represent with a vector index $\mathbf{p}^T = (p_1, p_2) $. From the definition, it follows that $\hat{D}_{\mathbf{p}}\hat{D}_{\mathbf{q}} = \tau^{<\mathbf{p},\mathbf{q}>}\hat{D}_{\mathbf{p+q}}$  where,  $\tau=-e^{i \frac{\pi}{n}}$ and $<\mathbf{p},\mathbf{q}> = p_2 q_1 - p_1 q_2$. The Weyl-Heisenberg group is not the only representation for the $\mathbb{Z}_n \times \mathbb{Z}_n$ group, but not all representations give rise to SIC-POVMs \cite{renes2004symmetric}. 

\begin{equation}\label{W-H group}
    \hat{D}_{\mathbf{p}}= \tau^{p_1 p_2} \hat{X}^{p_1}\hat{Z}^{p_1}
\end{equation}

The Weyl-Heisenberg group covariant SIC-POVMs are then defined as the orbit of a fiducial vector $|\psi_0\rangle$ over the Weyl-Heisenberg group, I.e., $\{\hat{D}_{\mathbf{p}}|\psi_0\rangle, \mathbf{p} \in \mathbb{Z}_n \times \mathbb{Z}_n\}$. For Weyl-Heisenberg group covariant solutions, the construction of a SIC-POVM reduces to solving for a fiducial vector $|\psi_0\rangle$ satisfying the set of polynomial equations $|\langle\psi_0| D_\mathbf{p}^{\dagger} D_\mathbf{q} |\psi_0\rangle|^2 = \frac{1}{n+1}$. The resulting set of equations are not trivial but can be simplified due to symmetries possessed by the Weyl-Heisenberg group. In 1999, Zauner \cite{zauner} showed that a fiducial vector, have a symmetry now known as Zauner's symmetry, can be used to simplify the polynomials and allow for an algebraic solution. The symmetry described by Zauner is that a fiducail vector exists in the degenerate subspace of the matrix shown in equation \ref{zauner matrix}.

\begin{equation}\label{zauner matrix}
    \langle j | Z | k \rangle = \frac{e^{i\lambda }}{\sqrt{n}}\tau^{2jk+j^2} , \lambda = \frac{\pi (d-1)}{12}
\end{equation}

The dimension of the eigenbasis of the Zauner matrix is $\lfloor (n+3-2k)/3 \rfloor$ for the eigenvalues $e^{i 2\pi k /3}$. The eigenvalue containing the fiducial vector in most known cases is $k=0$, but in some dimensions, solutions also exists in $k=1$ and $k=2$ \cite{scott2010symmetric,scott2017sics}. Zauner constructed solutions in dimensions 2-7 using the matrix \ref{zauner matrix}. Renes et al. \cite{renes2004symmetric} extended the solution list to dimension $45$. Later, Appleby \cite{appleby2005symmetric} showed that numerical Weyl-Heisenberg group covariant SIC-POVMs constructed by Renes et al.\cite{renes2004symmetric} are also symmetric to an order 3 Clifford unitary. Scott and Grassl \cite{scott2010symmetric} constructed the complete list of distinct SIC-POVMs based on the orbits they form on the extended Clifford group in all dimensions up to $50$. Scott \cite{scott2017sics} further extended the list of symmetries to all dimensions $\leq 121$ and some additional dimensions as high as $323$.  Currently, algebraic or exact solutions are known in the dimensions 2-21 and a hand full of higher dimensions as high as 323 \cite{exactfiducialsolns}. With spacial connection between SIC-POVMs of different dimensions allowing for the construction of analytic solutions in dimension 884 \cite{appleby2017dimension}.

Even with the known symmetries discovered so far, analytical construction of the fiducial vectors is a difficult problem, and often we have to relay on numerical methods to construct fiducial vectors. Numerical solutions have been constructed in every dimension up to dimension 151 and some dimensions as high as 1155  \cite{fuchs2010qbism}. The numerical method of construction is derived from the connection of the SIC-POVMs and spherical designs. The SIC-POVMs are a set of equally spaced vectors in the complex space that form a spherical 2-design \cite{renes2004symmetric}. The implication is that the SIC-POVMs form an equiangular tight frame, and they must satisfy the Welch bounds \ref{welch bound} \cite{welch1974lower}. The lower bound of which was computed by Benedetto and Fickus \cite{benedetto2003finite} using the frame potential \ref{fram potential}. The inequality \ref{welch bound} has $n^2(2n-1)$ parameters in its general form. By imposing group covariance, the free parameters reduces to only $2n-1$  which correspond to the fiducail vector. The number of free parameters is further reduced if additional symmetries exist in the fiducial vectors, such as the Zauner symmetry.

\begin{equation}\label{welch bound}
    \sum_{j,k}^{n^2} {|\langle \psi_j|\psi_k\rangle|^{2m}} \geq n^4 \frac{(n-1)!m!}{(n+m-1)!} , m \in \mathbb{Z}^{+}
\end{equation}

\begin{equation}\label{fram potential}
    Tr(S^2)=\sum_{p,q}^{n^2} {|\langle\psi_0| D_\mathbf{p}^{\dagger} D_\mathbf{q} |\psi_0\rangle|^4} \geq \frac{2 n^3}{n+1}
\end{equation}

If a fiducial vector is an eigenvalue of the  Zauner matrix \ref{zauner matrix}, then the number of free parameters is at most equal to the largest dimension of the degenerate subspaces of the Zauner matrix. This reduces the number of free parameters to $\lfloor (n+3)/3 \rfloor$ in most dimensions. 

The numerical method described above can generate SIC-POVMs without considering group covariance; however, many trials of random initial points are required to get a solution due to many local minima. For our search of general solutions without the condition of group covariance, we will introduce a different method which performs better in low dimensions in the following section.

\section{The Gram Matrices of SIC-POVMs}

The SIC-POVMs are a set of complex vectors that form a tight frame for the complex projective space. Thus, can be characterized using their Gram matrix.

The Gram matrix of a set of vectors $\{\vec{v}_k\}$ consists of the inner product of the vectors $(\vec{v}_j.\vec{v}_k)$. The Gram matrix is a symmetric matrix by definition. If the vectors from a tight frame, the corresponding Gram matrix is a projective matrix. The Gram matrix of a tight frame characterizes the frame uniquely up to a unitary operation, and therefore, we can construct the set of vectors given a Gram matrix up to a unitary equivalence. Thus, it is ideal for studying the equivalence classes of SIC-POVMs without any assumption of group covariance. 

Let the set of projective operators $E_k = \frac{|\psi_k\rangle \langle\psi_k|}{n}$ be a SIC-POVM, then the Gram matrix is given by equation \ref{sic gramian}. The Gram matrix of a SIC-POVM is a $n\times n$ symmetric matrix satisfying $Tr(P^k) = n$. This is because the SIC-POVMs are informationally complete (i.e., $\sum_k{E_k} =  \mathbb{I}$). The reconstruction of the SIC-POVM vectors from the Gram matrix can be found in appendix \ref{sic Gram reconstruction}.

\begin{equation}\label{sic gramian}
\begin{array}{ll}
    P &= \frac{\langle \psi_j | \psi_k \rangle}{n} \\
    &= \frac{1}{n\sqrt{n+1}} \left\{
    \begin{array}{cc}
        e^{i \phi_{jk}} &, j < k \\
        e^{-i \phi_{jk}} &, j > k \\
        \sqrt{n+1} &, j = k
    \end{array}\right.
\end{array}
\end{equation}

We will use the properties of the Gram matrix for the construction of SIC-POVMs. For this, let's define the space of matrices satisfying the conditions $Tr(P)=n$ and $Tr(P^2) = n$, having the form shown above, as  $\mathcal{B} = \{P, \phi_{jk} \in [0,2\pi]\}$. By definition, all SIC-POVM Gram matrices are elements of $\mathcal{B}$. If a matrix $P \in \mathcal{B}$ is a projective matrix, it will also be Gram matrix for a SIC-POVM. For a numerical search, it is more convenient to have a scalar function to check whether $P$ is projective, and for this we can use $Tr(P^k)$. In theorem \ref{trace check for sic gramians}, we show that we only need to show $Tr(P^3) = n$ and $Tr(P^4) = n$ to identify a projective matrix within the space $\mathcal{B}$. 

\begin{theorem}\label{trace check for sic gramians}
    Given a matrix $P \in \mathcal{B}$, the equations $Tr(P^3) = n$ and $Tr(P^4) = n$ are sufficient and enough conditions to show that the matrix $P$ is a projective matrix.
\end{theorem}

\begin{proof}
lets define two functions $f(x)$ and $g(x)$ as below, where $\lambda_k(\in \mathbb{R})$ are the eigenvalues of the matrix.

\begin{equation}\label{f and g}
    \begin{array}{ll}
    g(x)=\sum_{k}{\lambda_k^x}\\
    f(x)=\sum_{k}{|\lambda_k|^x}
    \end{array}
\end{equation}

We derive the following results from the definitions of the functions.
\begin{itemize}
    \item $\forall x\in \mathbb{R}^{+}, g(x)=Tr(P^x)$.
    \item $\forall x=2m, m\in \mathbb{N}, f(x)=g(x)$.
    \item $f(x) \geq \mathfrak{Re}[g(x)]$, Since $\forall \lambda_k\in \mathbb{R}, \mathfrak{Re}[\lambda_k^x] \leq |\lambda_k|^x$.
\end{itemize}

The function $f(x)$ is a convex function since $\frac{d^2}{dx^2}f=\sum_{k}{(ln|\lambda_k|)^2|\lambda_k|^x}\geq 0$. therefore, $f(a)=f(b)=f(c) \iff |\lambda_k|\in \{0,1\}$. 

For even powers, $Tr(P^2) = g(2) = f(2) = n$ and $Tr(P^4) = g(4) = f(4) = n$. For odd powers, since $f(x)$ is an upper bound to $g(x)$ and is a convex function, $g(3)\leq f(3)\leq n$. Therefore, if $g(3)=n$, then $f(3)=g(3)$. This shows that $\lambda_k$s are either $0$ or $1$. Finally we conclude that if $Tr(P^3) = n$ and $Tr(P^4) = n$ for a matrix $P \in \mathcal{B}$ if and only if $P$ is a projective matrix.
\end{proof}

We can now derive the equations $Tr(P^3) = n$ and $Tr(P^4) = n$, which we will use to construct SIC-POVMs without the assumption of group covariance. By directly expanding the functions $Tr(P^3)$ and $Tr(P^4)$, where the matrix $P$ takes the form shown in \ref{sic gramian}, we define the functions $\mathbf{f}_n(\Phi) = Tr(P^3)$ and $\mathbf{g}_n(\Phi) = Tr(P^4)$,  where $\Phi$ corresponds to phase parameters as follows:

\begin{equation}\label{tr 3}
    \mathbf{f}_n(\Phi) = \frac{3n-2}{n}+\frac{6}{n^3(n+1)^{\frac{3}{2}}}\sum_{a<b<c}^{n^2}{\cos{(\phi_{ab}+\phi_{bc}-\phi_{ac})}}
\end{equation}

\begin{equation}\label{tr 4}
\begin{array}{ll}
    \mathbf{g}_n(\Phi) = \frac{2 (n^3+2 n^2-n-1)}{n^2 (n+1)}  + \frac{24}{n^4(n+1)^{\frac{3}{2}}}\sum_{a<b<c}^{n^2}{\cos{(\phi_{ab}+\phi_{bc}-\phi_{ac})}}\\ + \frac{8}{n^4(n+1)^2}\sum_{a<b<c<d}^{n^2}{\bigg[\cos{(\phi_{ab}+\phi_{bc}+\phi_{cd}-\phi_{ad})}}\\  +\cos{(\phi_{ab}+\phi_{bd}-\phi_{cd}-\phi_{ac})} +\cos{(\phi_{ac}-\phi_{bc}+\phi_{bd}-\phi_{ad})}\bigg]
\end{array}
\end{equation}

We can think of the phase variables as elements of a $\frac{n^2(n^2-1)}{2}$ dimensional vector, and represent each element of $\mathcal{B}$ by a unique vector up to mod$(2\pi)$. We will refer to the set of phases as the phase vectors of the Gram matrices. For the vector representation of the phase variables, we define a single index $\chi(a,b) = (a-1) n^2-\frac{1}{2} a (a+1)+b $ where the indices $a$ and $b$ correspond to the index in the upper triangle of the Gram matrix. After this point, the phase vector $\Phi$ is represented using the index $\chi(a,b)$  where the value ranges from $1$ to $n^2(n^2-1)/2$.

The equations \fg are presented in a way that makes some symmetries of the Gram matrices obvious. That includes the continuous symmetry of $n^2-1$ dimension and the discrete symmetry to the permutation operation. These symmetries are needed to identify and characterize general SIC-POVM solutions generated using the functions $\mathbf{f}_n$ and $\mathbf{g}_n$.

\section{Continuity and relation of Group Covariant SIC-POVMs in the Gram Space}\label{symmetries of f and g}

We will start by identifying the continuous symmetry of the functions $\mathbf{f}_n(\Phi)$ and $\mathbf{g}_n(\Phi)$. To derive this, let's first define the vectors $\vec{K}_{rst}$ such that the phases inside the cosine functions can be written as a dot product $\vec{K}_{rst}\cdot\Phi$. In the function $\mathbf{f}_n$, the vectors $\vec{K}_{rst}$ have 3 non zero elements. For example, consider the first cosine functions in $\mathbf{f}_n$ which is  $\cos{(\phi_1+\phi_2-\phi_3)}$, then $\vec{K}_{1,2,3} = (1,-1,0,\ldots,0,1,0,\ldots,0)$ has elements $1$ at indices $\chi(1,2)$ and $\chi(2,3)$, -1 at the index $\chi(1,3)$ and $0$ everywhere else. We can then write the cosine function as $\cos{(\vec{K}_{1,2,3} . \Phi)}$. Similarly, we rewrite $\mathbf{f}_n$ as,

\begin{equation}
    \mathbf{f}_n(\Phi) = \frac{3n-2}{n}+\frac{6}{n^3(n+1)^{\frac{3}{2}}} \sum_{r,s,t} {\cos{(\vec{K}_{rst}}.\Phi)}
\end{equation}

The simplest symmetry one would expect from the equations is one where all the cosine functions in the sum remain fixed, and we can check for such a symmetry by simply checking for the linear independence of the vectors $\vec{K}_{rst}$. Since the vectors are $n^2(n^2-1)/2$ dimensional, it suffices to take $n^2(n^2-1)/2$ vectors. With a bit of combinatorics, it is possible to show that the vectors $\vec{K}_{rst}$ span a $(n^2-1)(n^2-2)/2$ space, the proof of which can be found in Appendix \ref{continuous symmetry of f and g}. This means, the functions \fg are invariant in the orthogonal space of the vectors $\vec{K}_{rst}$. We may also understand the continuous symmetry by noting that a unitary operation conserves all trace values and the only continuous unitary operations that form automorphism of $\mathcal{B}$ are $\Gamma=\{e^{i\Omega}, \Omega = \mbox{diag}(x_1,x_2,\ldots,x_{n^2}), x_k \in [0,2\pi]\}$. The unitary operations have $n^2$ parameters and directly translates to a $n^2-1$-dimensional subspace of $\mathcal{B}$. 

The second trivial symmetry we observe is the permutation symmetry. It is obvious the permutation operation is an automorphism of $\mathcal{B}$, and thus will be an invariant of the functions $\mathbf{f}_n(\Phi)$ and $\mathbf{g}_n(\Phi)$. We will represent the permutation operations with a permutation matrix $X_\sigma$. For a Gram matrix of a SIC-POVM, the permutation operations are equivalent to the reordering of the SIC-POVMs. The permutation group generates a large set of Gram matrices which will become important in classifying numerical SIC-POVMs. 

The two trivial transformations described above are both symmetries of the functions and they map the space of matrices $\mathcal{B}$ to its self.  For the purpose of generating general SIC-POVMs however, we don't need the transformations to be automorphism of the entire matrix space $\mathcal{B}$. We only need the transformation to map one SIC Gram matrix to another. One such a transformation we can search for is a continuous transformation that preserves the functions $\mathbf{f}_n(\Phi) = n$ and $\mathbf{g}_n(\Phi) = n$. In other words, are the SIC-POVM Gram matrices path connected in the space of $\Phi$? Through numerical method it has been shown that for dimensions $4-16$, the known Weyl-Heisenberg group covariant SIC-POVMs are disjoint up to $\Gamma$ \cite{bruzda2017quantum}. Using the trace functions, the continuity of SIC-POVM Gram matrices corresponds to the intersection space of the surfaces $\mathbf{f}_n(\Phi) = n$ and $\mathbf{g}_n(\Phi) = n$. The first immediate result is that SIC-POVM Gram matrices exist on a critical point of both functions. This is an expected result at least in the dimensions 4-16 where \wh covariant SIC-POVMs are known to be isolated in \cite{bruzda2017quantum}. In theorem \ref{SICs on critical point}, we show that this is a general result that holds in arbitrary dimensions.

\begin{theorem}\label{SICs on critical point}
    SIC-POVM Gram matrices are located at the common critical points of the surfaces $\mathbf{f}_n(\Phi) = n$ and $\mathbf{g}_n(\Phi) = n$.
\end{theorem}

\begin{proof}
    The straightforward method to prove the theorem would be to construct the gradients of the two functions and confirm that both gradients are 0. However, this would be an unnecessarily long method. Instead, we will reconstruct the functions by considering an infinitesimal shift of the vector $\vec{\Phi}$.

    Let $G$ be the Gram matrix of a SIC-POVM of $n$-dimensional Hilbert space, and the shift on $\vec{\Phi}$ be $\vec{v}$. We can then write the new Gram matrix as,

    \begin{equation}
        g_{jk}(\vec{\Phi}+\vec{v}) = \frac{e^{i\phi_{jk}+i\lambda v_{jk}}}{n\sqrt{n+1}}.
    \end{equation}
    For $|\vec{v}| << 1$, we expand the exponential to the first order in $\vec{v}$ and approximate the Gram matrix as follows.

    \begin{equation}
        g_{jk}(\vec{\Phi}+ \vec{v}) \simeq \frac{e^{i\phi_{jk}}}{n\sqrt{n+1}} + i v_{jk}\frac{e^{i\phi_{jk}}}{n\sqrt{n+1}}
    \end{equation}

    \begin{equation}
        G(\vec{\Phi}+\vec{v}) \simeq G(\vec{\Phi}) + \Delta
    \end{equation}

    We can now compute the trace cube and fourth power and approximate it to the first order in the shift $\vec{v}$. Since the Gram matrices of SIC-POVMs are projective matrices, we can reduce all powers of $G(\vec{\Phi})^k$ to $G(\vec{\Phi})$ greatly reducing the calculation.
    \begin{equation}\label{grad f}
    \begin{array}{ll}
        Tr(G(\vec{\Phi}+\vec{v})^3) &\simeq Tr(G(\vec{\Phi})^3+3 G^2 \Delta + O(\Delta^2)) \nonumber \\
        & \simeq Tr(G(\vec{\Phi})+3 G \Delta + O(\Delta^2)) \nonumber \\
        & \simeq n +3 Tr(G(\vec{\Phi}) \Delta)
        \end{array}
    \end{equation}
Similarly,
    \begin{equation}\label{grad g}
    \begin{array}{ll}
    Tr(G(\vec{\Phi}+\vec{\delta})^4) &\simeq Tr(G(\vec{\Phi})^4+4 G^2 \Delta + O(\Delta^2)) \nonumber \\ 
    & \simeq n +4 Tr(G(\vec{\Phi}) \Delta)
    \end{array}
    \end{equation}

    Note that the first order correction on both functions depends on the term $Tr(G(\vec{\Phi}) \Delta)$. This shows that if the first order correction of both functions vanishes simultaneously. We then expand the first order correction in to its elements and group the terms based on $\delta_{jk}$.

    \begin{equation}\label{gradiant zero}
    \begin{array}{ll}
        Tr(G(\vec{\Phi}) \Delta) &= \sum_{ab}{i \frac{e^{i\phi_{ab}}}{n\sqrt{n+1}}  v_{ba}\frac{e^{i\phi_{ba}}}{n\sqrt{n+1}}} , \mbox{ where } a\neq b \nonumber \\
        &= \frac{i}{n^2(n+1)} \sum_{ab}{e^{i(\phi_{ab}+\phi_{ba})} v_{ba}}\nonumber \\
        &= 0
        \end{array}
    \end{equation}

    Since the first order correction in equations \ref{grad f} and \ref{grad g} is equal to $\vec{v} \cdot \nabla \mathbf{f}_n$ and  $\vec{v} \cdot \nabla \mathbf{g}_n$ respectively, and vanished for both functions, all SIC-POVM Gram matrices exist on critical points of both functions.
    
\end{proof}

The SIC-POVM Gram matrices exist on critical points of the functions \fg, which is a necessary condition for the solutions to be isolated. However, this does not imply that the solutions are isolated. To achieve that, we need to generate the Hessian matrices of the two functions. The null space of Hessian matrices determine the continuous subspace of SIC-POVM Gram matrices. To compute the Hessian matrix, we simply expand the Gram matrix to the second order. The trace functions to the second order correction take form $ G(\vec{\Phi}+\vec{v}) \simeq G(\vec{\Phi})+ G^2 \Delta^{(1)} + G \Delta^{(2)} + O(\Delta^2))$, where $\Delta^{(1)} = i e^{i \phi_{ab}}v_{ab}/n\sqrt{n+1}$ and $\Delta^{(2)} = - e^{i \phi_{ab}}v_{ab}^2/n\sqrt{n+1}$. Then the trace of \fg can be expanded to the second order, from which we generate the Hessian matrices.

\begin{equation}\label{grad f to 2nd order}
        Tr(G(\vec{\Phi}+\vec{v})^3) \simeq n + 3 (Tr(G \Delta^{(2)}) + Tr(G \Delta^{(1)} \Delta^{(1)})) + O(\Delta^3)) 
\end{equation}

 \begin{equation}\label{grad g to 2nd order}
    \begin{array}{l}
        Tr(G(\vec{\Phi}+\vec{v})^4) \simeq n + 4 (Tr(G \Delta^{(2)}) + Tr(G \Delta^{(1)} \Delta^{(1)})) \\ + 2 Tr(G \Delta^{(1)}G \Delta^{(1)}) + O(\Delta^3))
    \end{array}
\end{equation}

The second order correction of a functions is given by the Hessian matrix as $\frac{1}{2} H_{jk} x_j x_k$. In equations \ref{grad f to 2nd order} and \ref{grad g to 2nd order}, the variables are the vector elements $v_{ab}$, and by extracting the terms of the sum based on the vector elements, we generate the Hessian matrices of the two functions. Unlike the gradient case, the Hessian matrices are difficult to simplify, and generating a general result of continuity is challenging. As is evident from expressions, the Hessian matrices involve the triple products of the SIC-POVM, and general proof of continuity requires further restrictions on these products. However, numerically, we can still use the expression to explore local structure of the surfaces defined by \fg. In dimension 3, the intersection of the null spaces of the two matrices form a $10$-dimensional space. This means that two free parameters exist which describe the continuous fiducial vectors that have been constructed in dimension $3$. In dimensions $4-7$, the intersection of the null space of the two Hessian matrices form a $n^2-1$ dimensional space, which is the continuous symmetry $\Gamma$ introduced above. We explore the local structures of general solutions by using the numerical solutions in the next section. 

\section{Characteristics of General Numerical solutions}\label{numerical sics}

Numerical construction of Weyl-Heisenberg group covariant SIC-POVMs is achieved by using the inequality \ref{fram potential}. The search would typically require many random initial points because the function contains local minima. In principle the Welch bounds \ref{welch bound} can be applied to construct SIC-POVMs without the assumption of group covariance. However, practically, due to the many local minima, the method becomes impractical. Instead, we will search for a valid SIC-POVM Gram matrix by searching for the intersection of the functions \fg as shown in equations \ref{tr 3} and \ref{tr 4}. This increases the free parameters to $O(n^4)$, but in dimensions $\leq 5$, we don't encounter any local minima and in dimensions $6$ and $7$, the frequency of local minima to the global is manageable. We introduce a new functions $ \mathbf{S}( \Phi) = (f(\Phi)-n)^2+(g(\Phi)-n)^2 \geq 0$ to find points satisfying both conditions simultaneously. The method of search used is conjugate gradient on the software Wolfram Mathematica. In order to speed up the search, the SIC-POVMs we generated are set to accuracy of $10^{-8}$ which is enough for the characterisation of the solutions. To confirmation that the solutions are indeed SIC-POVMs and not a very close approximate, we refined the solutions to $100$ digits. 

From the matrices $\in \mathcal{B}$, we know that each solution have a trivial symmetry generated by the diagonal unitary operations. The results of Bruzda eq al. \cite{bruzda2017quantum}, which we have also found using the trace functions, show the existence of a connected subset of SIC-POVM Gram matrices of $n^2-1$-dimensional space in dimensions 4-7. These "islands" of SIC-POVM Gram matrices are further connected though the permutation operations. We say a two solutions are distinct from each other if they are not in the same island of Gram matrices and are on different orbits of the permutation operators.

The islands of SIC-POVM Gram matrices ($\in \mathcal{B}$) can easily be described in the vector space of $\vec{\Phi}$ as $n^2-1$-dimensional planes, as shown in section \ref{symmetries of f and g}. The Bargmann invariants of the SIC-POVMs are the ideal tools to check if two SIC-POVMs \cite{appleby2011lie} belong to the same island. Let's explain the structure of these islands and briefly discuss how to group solutions using the Bargmann invariants. The islands of solutions are formed by the diagonal unitary operations on the $n^2$ dimensional Hilbert space. Since the Gram matrices of two SIC-POVMs equivalent up to a unitary operator are identical, a global phase shift on all the SIC-POVMs does not generate a different Gram matrix. This means that the islands are generated by unitary operators which are elements of $SU(n^2)$. Thus, the operators can be represented by a set of $n^2-1$ free parameters. Let the diagonal unitary operator $U$ be $U = e^{i\Omega}$ where $\Omega = \mbox{diag}(c_1,c_2,\ldots,c_{n^2})$, and  $\sum_k {c_k}=0$. Applying the operator on a $P\in \mathcal{B}$ as $P' = U^\dagger P U$, we generate a different Gram matrix in the same island as $P$.

\begin{equation}
    P' =\frac{1}{n\sqrt{n+1}}\left\{
    \begin{array}{ll}
        e^{i (\phi_{jk}-c_j+c_k)} &, j < k \\
        e^{-i (\phi_{jk}-c_j+c_k)} &, j > k \\
        \sqrt{n+1} &, j = k
    \end{array}\right.
\end{equation}

This corresponds to a shift in the phase vector as $\phi_{jk} \mapsto \phi_{jk} - c_j +c_k$, forming the $n^2-1$-dimensional plane. The Bargmann invariants of a SIC-POVM are given by the trace of triple products of the  $T_{ijk}= n^3 Tr(E_i,E_j,E_k)$. If we generate the Bargmann invariant of the Gram matrix $P'$, we see that the variables $c_k$ disappear, and the entire island has unique $T_{ijk}$s. Conversely, if two SIC-POVMs have identical Bargmann invariants, then fixing one of the indices gives us two identical Gram matrices, which indicates that the two SIC-POVMs belong to the same island \cite{appleby2011lie}. By using the Bargmann invariants, we filter all generated solutions so that all belong to different islands.

The second symmetry of the functions \fg is due to the permutation operations $\sigma$. Not all permutation operations however connect different islands as shown in theorem \ref{permutation op. automorphic}. Compared to the total number of permutations possible, the number of automorphisms is much less, and the total number of islands is $(\approx n^2!)$. 

\begin{theorem}\label{permutation op. automorphic}
    A permutation operation $X_\sigma$ that form an automorphism of an island of Gram matrices if and only if there exists a projective unitary operations $V$ and SIC-POVM vectors $\{|\psi_k\rangle\}$ such that,
    \begin{equation}
        V|\psi_k\rangle \mapsto |\psi_{\sigma(k)}\rangle \mbox{ , } \forall k\in \mathbb{Z}_n\times\mathbb{Z}_n
    \end{equation}
    where $\sigma$ is the permutation of indices.
\end{theorem}

\begin{proof}
   Let $P$ and $P'$ be SIC-POVM Gram matrices, where $P' = X_\sigma^\dagger P X_\sigma$ and the two Gram matrices belong to the same island of solutions. We then construct two sets of SIC-POVMs $|\psi_k\rangle$ and $|\psi'_k\rangle$ from the Gram matrices, respectively. Since both Gram matrices belong to the same island of solutions, there necessarily exist projective unitary matrix that maps one to the other.

   \begin{equation}
       |\psi'_k\rangle = U |\psi_k\rangle 
   \end{equation}

   On the other hand, we can reconstruct the Gram matrix we have by simply applying the permutation on the SIC-POVM: $|\psi_k\rangle \mapsto |\psi_{\sigma(k)}\rangle$. At this point, we have two SIC-POVMs that generate an identical Gram matrix and hence are unitarily equivalent to each other as $|\psi_{\sigma(k)}\rangle = C |\psi'_k\rangle$, where $C \in U(n)$.
   
   We can then construct a projective unitary operation that maps the $|\psi_k\rangle$ to $|\psi_{\sigma(k)}\rangle$.

   \begin{equation}
       |\psi_{\sigma(k)}\rangle = V |\psi_k\rangle \mbox{ , where } V = CU
   \end{equation}
   
\end{proof}

Permutations that generate automorphisms indicate the symmetries possessed by the SIC-POVMs in a given island. For the purpose of describing the group-covariant SIC-POVM islands, we will focus on automorphisms that have fixed points in a given island of SIC-POVM Gram matrices. 

The Gram matrices constructed with the functions \fg can be an element of any island of solution, and checking the equivalence of any two Gram matrices is difficult due to the large number of possible permutations. All permutations that do not form an automorphism of a given island, generate SIC-POVMs that are not unitarily equivalent to the initial SIC-POVMs. To identify if two Gram matrices are equivalent, one can use the Gram matrices. Before explicitly constructing the permutation operations connecting two islands, we will define the generating set of the Bargmann invariant.

\begin{define}\label{generating set}
    The generating set of the Bargmann invariant $gen[T_{ijk}]$ is a set containing all the phases that appear in the tensor of the Bargmann invariants i.e., $T_{ijk}=Tr(\Pi_i\Pi_j\Pi_k)$.
\end{define}

If two SIC-POVMs have different generating set, then the two Gram matrices are not equivalent up to unitary operation. This is because two SIC-POVMs are equivalent up to a projective unitary transformation if and only if their Bargmann invariants are identical \cite{appleby2011lie}. The generating set is the first criteria to determine equivalence of two SIC-POVMs. However, if the generating set of two SIC-POVMs is the same, the two SIC-POVMs are not necessarily equivalent. therefore, for solutions that do have a common generating set, we need to explicitly construct the permutation operation in order to show their equivalence. We will use the Gram matrix instead of the tensor of Bargmann invariants for the construction of the permutation operations. 

In order to compare the Gram matrices belonging to different island of solutions, we first define a means of identifying the islands in a way that allows us to construct a permutation operations between solutions. The best choice for this is to shift the phases of a given row and column to 0. We can do this by using the continuous symmetry of the islands.

\begin{define}
Let $\Phi$ be an element of a SIC island. Then apply the phase shift $\Gamma_{\lambda_k}(\Phi)$ such that, $\Phi \mapsto \Phi_a$ where $ \Phi_a =\{(\phi_{j,k}+(\phi_{a,j}-\phi_{a,k})),j<k \}$.
\end{define}

By reducing the 1st row and column phases to 0 (i.e., $\Phi_1$) and taking the modulus of all phases, we uniquely identify an island of SIC-POVMs. The advantage of such a phase shift is that all entries of the Gram matrix will reduce to elements of the generating set, as is evident from the form of $\Phi_0$ above. We will briefly discuss the symmetries of the islands in dimensions 4-7 one by one.

\subsection{Dimension 4}
In dimension 4, the functions \fg yield general solutions very quickly despite the large number of parameters. From the numerical solutions, we observe that all generated solutions have the generating set \ref{gennum4}, indicating that all the solutions may be equivalent up to permutation. Notice that the generating set is the same if we multiply all the phases by $-1$, which means the set also characterizes anti-unitarily equivalent SIC-POVMs.

\begin{equation}\label{gennum4}
    \begin{array}{l}
    gen[T_{ijk}] = \{0,0.33312,0.571437,0.666239,0.904557,\\0.999359,1.5708, 1.90392,2.47535,3.80783,4.37927,\\4.71239,5.28383,5.37863,5.61695,5.71175,5.95007\}
    \end{array}
\end{equation}

After projecting the Gram matrices onto a unique point in their respective island, the generating phases are distributed in the matrices with the respective frequencies given below.

\begin{equation}\label{freqgennum4}
   \{49,18,9,18,9,18,18,9,9,18,18,18,18,6,6,6,6\}
\end{equation}

To construct permutation operation from two islands, we start with the indices of the lowest frequency and work our way up. As a result, for all numerical solutions constructed ($>2\times 10^5$), a unique permutation was constructed, mapping each solution to a Weyl-Heisenberg group-covariant solution. 

Projecting the island on to a single point $\Phi_a$ allows us to also generate an automorphism $Aut(\Phi_a)=Y_a$ that fixes the $a^{th}$ index. By going through all the possible values of $a$, we construct $16$ automorphisms. Each of the automorphisms are formed by cycles of three indices and one fixed index, meaning the automorphisms are all of order-3, $Y_a^3 = \mathbb{I}$. The matrices $Y_a$ generate a group of automorphisms $<Y_a>$, which has 48 elements, which contains a subgroup of $H$ with $16$ elements, where $h_k^4 = \mathbb{I}, \forall h_k \in H/\mathbb{I}$. 

The matrices $Y_a$ indicate that the SIC-POVMs in the given island can be found as eigenvalues of unitary which fixes the set of vectors of a given SIC-POVM. The order-4 matrices of the group $H$, indicate that the SIC-POVMs of the island are covariant with respect to a group isomorphic to $H$. If the island is equivalent to a Weyl-Heisenberg group covariant SIC-POVMs, the group $H$ is isomorphic the Weyl-Heisenberg group, and the elements $Y_a$ correspond to order-3 unitaries such as the Zauner matrix and other similar unitaries. 

\subsection{Dimension 5}

In dimension 5, the numerical solutions converge relatively fast without for all arbitrary initial search points used. This allowed us to generate more than $10^4$ numerical solutions. Similar to dimension 4, all solutions constructed have the generating set \ref{gennum5} with 73 elements. 

\begin{equation}\label{gennum5}
\begin{array}{ll}
    gen[T_{ijk}] = \{0,0.00220427,0.00903612,0.0460726,\\
    0.123129,0.23881,0.295211,0.321435,0.434237,\\
    0.487521,0.498761,0.536226,0.675251,0.720412,\\
    0.757876,0.885633,0.905025,1.02595,1.10858,1.11982,\\
    1.24631,1.25884,1.30271,1.37977,1.45798,1.48732,\\
    1.55185,1.60825,1.62764,1.63447,1.93189,2.02575,\\
    2.07904,2.13544,2.36521,2.6501,2.75208,3.5311,\\
3.63309,3.91797,4.14775,4.20415,4.25743,4.3513,\\
4.64871,4.65554,4.67494,4.73134,4.79586,4.8252,\\
4.90342,4.98048,5.02434,5.03688,5.16337,5.17461,\\
5.25723,5.37816,5.39755,5.52531,5.56277,5.60793,\\
5.74696,5.78442,5.79566,5.84895,5.96175,5.98797,\\
6.04437,6.16006,6.23711,6.27415,6.28098\}
\end{array}
\end{equation}

The corresponding frequency is given in  \ref{freqgennum5}. Similar to the case of dimension 4, we generating permutation matrices by starting with the lowest frequency and work our way up. For all constructed Gram matrices, we were able to construct a unique permutation that maps the numerical solutions to a Weyl-Heisenberg group covariant SIC-POVM Gram matrix.

\begin{equation}\label{freqgennum5}
    \begin{array}{ll}
    \{73,9,9,9,9,9,9,9,9,9,3,9,9,9,9,9,3,9,9,9,\\
    9,9,9,9,9,9,9,9,9,9,9,9,9,9,9,9,9,3,3,9,9,9,\\
    9,9,9,9,3,9,3,9,9,9,9,9,9,9,3,3,9,9,9,9,9,9,\\
    3,3,9,3,3,3,3,3,3\}
    \end{array}
\end{equation}

The automorphisms of the islands can be constructed in a similar manner to dimension 4. The group $<Y_a>$ has 75 elements where all the $Y_a^3=\mathbb{I}$ and the subgroup $H$ contains elements satisfying $h_k^5=\mathbb{I}$. Both matrices correspond to the order-3 unitaries fixing a fiducial vector and the Weyl-Heisenberg group in dimension 5. 

\subsection{Dimensions 6 and 7}
In dimensions 6 and 7, in addition to the difficulty of obtaining numerical solutions due to the large number of free parameters, we also encounter local minima. This makes the time needed to generate a numerical solution long. With that being said, we constructed more than $100$ solutions  in dimension 6 and ($20$) solutions in dimension 7. 

The numerical SIC-POVMs generated in these dimensions are all equivalent to the Weyl-Heisenberg group-covariant SIC-POVMs. The corresponding generating set is too big to include here, but since they are equivalent to Weyl-Heisenberg group solutions one can easily form them from the analytic solutions or numerical solutions found in \cite{fuchs2010qbism}. 

The automorphisms of the islands in dimension 6, unlike the previous two dimensions are not $n^2$ in number. In dimensions $n=0\mbox{ } mod \mbox{ } 3$, automorphisms with a single fixed index can not be formed. Instead we have matrices that fix 3 indices per automorphism, forming 12 distinct elements. We can use these to construct the 36 matrices corresponding to  order-6 automorphisms, i.e., isomorphic to Weyl-Heisenberg group in dimension 6.

Dimension 7 Gram matrices similarly allow automorphisms of order-3 and order-7 permutation matrices with 147 elements. Similar to all previous cases, the order-7 permutations are isomorphic to the Weyl-Heisenberg group and the order-3 permutation matrices correspond to the symmetries of the fiducial vectors.

\section{Conclusion}

By using the trace functions \fg, we conclude the following two points. The first and easier conclusion is that the trace functions do not have non-trivial symmetry that generates new island of SIC-POVM Gram matrices. This is because, the numerical solutions depend on the initial point we chose, and any solutions that are related by the symmetry of the two functions would be generated with equal likelihood. For the sake of argument, let's say there is an order-two unitary matrix that fixes the trace functions. This would mean that for every solution generated so far, there is a second Gram matrix satisfying the trace conditions increasing the total number of solutions by a factor of 2. The probability that we only found one of the two types of solutions becomes $1/2^N$ for $N$ trials. The number of solutions we have generated is large enough to reduce the probability to orders less than $10^{-30}$.

The second conclusion is that the SIC-POVM Gram matrices allow automorphism groups. When we consider solutions that are not equivalent up to a symmetry of the trace functions, there is no reason to assume that they would be generated with equal likelihood in the numerical search. In fact, the larger the size of automorphism group of solution islands, the less likely the solutions are to be generated through numerical search. All the solutions we have generated exhibit the Weyl-Heisenberg group symmetry, meaning the automorphism group is isomorphic to the Clifford group which has $n^5\prod_{p|n}(1-p^{-2})$ elements \cite{scott2010symmetric}. If a solution exists having an automorphism group of smaller size, it would be more likely to be found through numerical search, simply due to the larger number of solutions that would exist. This suggests that SIC-POVMs Gram matrices must allow for automorphism groups making them all group covariant. In all the solutions we generated, we see that this group is the Weyl-Heisenberg group. It is important to note that the results don't allow us to rule out solutions that have larger symmetry group, such as to the Hoggar SIC-POVMs in dimension 8, which have a much smaller generating set.

The real space analog of SIC-POVMs is the simplex, which forms a unique and  highly symmetric Gram matrix. We don't have such a symmetry in the SIC-POVM Gram matrix; however from the numerical results in the dimensions 4 and 5, we see that the generating set of SIC-POVMs is unique, and all Gram matrices are equivalent up to a permutation operation. The apparent automorphisms of all the constructed the islands of Gram matrices need an analytic proof as it seems to be not just a method of construction but rather a necessity for the existence of SIC-POVMs.  

\section*{Acknowledgments1}
Part of this work was carried out at Bilimler Köyü at Foça.

\appendix

\section{Construction of a SIC-POVM Elements from the Gram Matrix}\label{sic Gram reconstruction}

Consider the SIC-POVM in dimension $n$, given by the set of equiangular vectors $\{|\nu_k\rangle,\langle\nu_i|\nu_k\rangle=\frac{n\delta_{ik}+1}{n+1}\}$. Let's first define the map,

\begin{equation}
    V:\mathbb{C}^{n^2}\mapsto \mathbb{C}^n
\end{equation}
where $V=(\frac{1}{\sqrt{n}}|\nu_1\rangle,\frac{1}{\sqrt{n}}|\nu_2\rangle,\ldots,\frac{1}{\sqrt{n}}|\nu_{n^2}\rangle) \in \mathbb{C}^{n \times n^2}$. Since the vector $|\nu_k\rangle$ span the complex space $\mathbb{C}^n$, any vector in the $n$-dimensional complex space can written as,

\begin{equation}
    |\psi\rangle=V.\left(\begin{array}{l}
x_1\\ \vdots \\x_{n^2}
\end{array}\right)=\sum_k{x_k\frac{1}{\sqrt{n}}|\nu_k\rangle}.
\end{equation}

We also know that the equiangular vectors form a tight frame; therefore, $V V^\dagger=\sum_{k=1}^{n^2} \frac{|\nu_k\rangle\langle\nu_k|}{n} =\mathbb{I}_{n\times n}$. Using the map $V$, we rewrite the Gram matrix of the SIC-POVM as, $P=V^\dagger V \in \mathbb{C}^{n^2\times n^2}$.

\begin{equation}
    P=V^\dagger V=\frac{1}{n}\sum_{i,k}^{n^2}{\langle \nu_i|\nu_k \rangle}
\end{equation}

We can immediately see the following properties of the Gram matrix $P$:
\begin{enumerate}
    \item The Gram matrix P is a Hermitian, since $P^\dagger=(V^\dagger V)^\dagger = V^\dagger V=P$
    \item $P$ is a projective matrix (i.e., $P^2 = (V^\dagger V)( V^\dagger V)=V^\dagger (V V^\dagger) V = V^\dagger V =P$)
    \item Since $P$ is a projective matrix, it's rank can be computed from its trace.\\ I.e., $rank(P)=Tr(P)=\frac{1}{n}\sum_{k}^{n^2}{\langle \nu_k|\nu_k \rangle}=Tr(V V^\dagger)=Tr(\mathbb{I})=n$.
\end{enumerate}

To construct the set of vectors forming a SIC-POVM from a SIC Gram matrix, we use the fact that  $V V^\dagger = \mathbb{I}$ and $P V^\dagger = V^\dagger$. As a result, the row vectors of $V$ form a projection basis for the Gram matrix. We use the Gram-Schmidt process to construct the basis of the projection subspace. Using the basis vectors as the row vectors we form the $n\times n^2$ matrix $V$. The column vectors of the matrix $V$ are SIC-POVM elements. 

\section{Continuous Symmetry of the $Tr(P^3)$ and $Tr(P^4)$ Functions}\label{continuous symmetry of f and g}

The continuous symmetry of the functions \fg corresponds to the invariance of individual cosine functions in the functions \fg to a translation on the phase vector $\mathbf{\Phi}$. Lets define the vectors $\vec{K}_{abc}$ for $a<b<c$ such that $\cos (\vec{K}_{abc}.\mathbf{\Phi}) = \cos (\phi_{ab}+\phi_{bc}-\phi_{ac})$. The continuous symmetry appears because the vectors $K_{abc}$ do not span a $n^2(n^2-1)/2$ dimensional space.

Let's define a basis formed with the vectors $\mathbf{K} = \{K_{ijk} :i<j<k , j=i+1 $\}. We will then show that any vector $\vec{K}_{abc}$ belongs to the space spanned by $\mathbf{K}$. For simplicity of further calculations, the vector $\vec{K}_{abc}$ is presented with its non zero indices as $(a,b)-(a,c)+(b,c)$.

Given indices $a$, $b$ and $c$ the phases $(\phi_{ab}, \phi_{bc}, \phi_{ac})$, we will start by adding elements of the basis $\mathbf{K}$ incrementally until we generate $\vec{K}_{abc}$.

\begin{equation}
    \begin{array}{l}
      (a,b) - (a,c) + (b,c) \rightarrow -((a,a+1)-(-a,b)+(a+1,b)) \\ +((a,a+1)-(-a,c)+(a+1,c))+(b,c)
    \end{array}
\end{equation}

Notice that the indices $(a+1,b)$ and $(a+1,c)$ are not part of the $\vec{K}_{abc}$, so we will add vectors from $\mathbf{K}$ to cancel the two elements. which will leave two terms, $(a+2,b)$ and $(a+2,c)$ that we need to removed. We repeat the process until $a+m = b-1$, as shown below.

\begin{equation}\label{continuous symmetry vector expansion}
    \begin{array}{ll}
      (a,b) - (a,c) + (b,c) \rightarrow -\bigg((a,a+1)-(a,b)+(a+1,b)\bigg)\\
      +\bigg((a,a+1)-(a,c)+(a+1,c)\bigg)\\
      - \bigg((a+1,a+2)-(a+1,b)+(a+2,b)\bigg)\\
      + \bigg((a+1,a+2)-(a+1,c)+(a+2,c)\bigg) \\
      \vdots\\
      + \bigg((b-1,b)-(b-1,c)+(b,c)\bigg)
    \end{array}
\end{equation}

Note that every index is canceled out in equation \ref{continuous symmetry vector expansion}, and only the term $(a,b)-(a,c)+(b,c)$ remains. This shows that, for arbitrary indices $a,b,$ and $c$, the vector $\vec{K}_{abc}$ is an element of the space spanned by $\mathbf{K}$. 

The phases in the function $\mathbf{g}_n$ can be broken down into sums of the vectors $\vec{K}_{ijk}$ as well. As shown above, all such vectors are elements of the space spanned by $\mathbf{K}$.

\begin{equation}
\begin{array}{l}
      (a,b) - (a,d) + (b,c) + (c,d) \rightarrow \bigg((a,b) - (a,c) + (a,d)\bigg)\\
      + \bigg((a,c) - (a,d) + (c,d)\bigg)
      \end{array}
\end{equation}

\begin{equation}
\begin{array}{l}
      (a,b) - (a,c) + (b,d) - (c,d) \rightarrow \bigg((a,b) - (a,c) + (b,c)\bigg)\\
      - \bigg((b,c) - (b,d) + (c,d)\bigg)
      \end{array}
\end{equation}

\begin{equation}
\begin{array}{l}
      (a,c) - (a,d) - (b,c) + (b,d) \rightarrow \bigg((a,b) - (a,d) + (b,d)\bigg)\\
      - \bigg((a,b) - (a,c) + (b,c)\bigg)
      \end{array}
\end{equation}

The orthogonal vectors to the space spanned by $\mathbf{K}$, form  the continuous symmetry of the functions \fg. We can then show the dimension pf the symmetry as the difference in the dimension of the space $\vec{\mathbf{\Phi}}$ and the space spanned by $\mathbf{K}$ as, $n^2(n^2-1)/2-(n^2-1)(n^2-2)/2$ which is $(n^2-1)$.

\section*{Reference}
\bibliography{ref}

\end{document}